\documentclass[sigconf,review=false,screen,10pt]{acmart} 
\pdfoutput=1

\usepackage{listings}
\usepackage{url}
\usepackage{hyperref}
\usepackage{todonotes}
\usepackage{menukeys}

\newcommand{\var}[1]{\text{\lstinline+#1+}}
\newtheorem{theorem}{Theorem}

\newtheorem{lemma}{Lemma}
\newtheorem{informal}{Informal Axiom}
\newtheorem{axiom}{Axiom}

\newcommand{\ba}{\ensuremath{{\boldsymbol a}}}   
\newcommand{\bb}{\ensuremath{{\boldsymbol b}}}    
\newcommand{\bc}{\ensuremath{{\boldsymbol c}}}    
\newcommand{\bd}{\ensuremath{{\boldsymbol d}}}   
\newcommand{\br}{\ensuremath{{\boldsymbol r}}}   
\newcommand{\bt}{\ensuremath{{\boldsymbol t}}}   
  
\newcommand{\bv}{\ensuremath{{\boldsymbol v}}}   
\newcommand{\bw}{\ensuremath{{\boldsymbol w}}}
\newcommand{\bx}{\ensuremath{{\boldsymbol x}}}
\newcommand{\by}{\ensuremath{{\boldsymbol y}}}
\newcommand{\bz}{\ensuremath{{\boldsymbol z}}}
\newcommand{\bone}{\ensuremath{{\boldsymbol 1}}}
\newcommand{\bzero}{\ensuremath{{\boldsymbol 0}}}

\newcommand{\set}[1]{\left\{ #1 \right\}}

\renewcommand{\int}{\ensuremath{\mathit{int}}}
\newcommand{\upper}{\ensuremath{\mathit{upper}}}
\newcommand{\epi}{\ensuremath{\mathit{epi}}}

\newcommand{\eqnlabel}[1]{\label{eq:#1}}
\newcommand{\nakedeqnref}[1]{\ref{eq:#1}}
\newcommand{\eqnref}[1]{Equation~\nakedeqnref{#1}}

\newcommand{\lemmalabel}[1]{\label{lemma:#1}}
\newcommand{\nakedlemmaref}[1]{\ref{lemma:#1}}
\newcommand{\lemmaref}[1]{Lemma~\nakedlemmaref{#1}}

\newcommand{\thmlabel}[1]{\label{thm:#1}}

\newcommand{\seclabel}[1]{\label{sec:#1}}
\newcommand{\nakedsecref}[1]{\ref{sec:#1}}
\newcommand{\secref}[1]{Section~\nakedsecref{#1}}

\newcommand{\Reals}{\ensuremath{\mathbb{R}}}
\newcommand{\PosReals}{\ensuremath{\mathbb{R}_{>0}}}
\newcommand{\NNegReals}{\ensuremath{\mathbb{R}_{\geq 0}}}
\newcommand{\NegReals}{\ensuremath{\mathbb{R}_{<0}}}

\graphicspath{{./pictures/}}

\copyrightyear{2021}
\acmYear{2021}
\setcopyright{acmlicensed}\acmConference[AFT '21]{3rd ACM Conference on Advances in Financial Technologies}{September 26--28, 2021}{Arlington, VA, USA}
\acmBooktitle{3rd ACM Conference on Advances in Financial Technologies (AFT '21), September 26--28, 2021, Arlington, VA, USA}
\acmPrice{15.00}
\acmDOI{10.1145/3479722.3480987}
\acmISBN{978-1-4503-9082-8/21/09}

\begin{document}

\title{Composing Networks of Automated Market Makers}


\author{Daniel Engel}
\email{daniel_engel1@brown.edu}
\affiliation{%
  \institution{Brown University Computer Science Department}
}

\author{Maurice Herlihy}
\email{maurice.herlihy@gmail.com}
\affiliation{%
  \institution{Brown University Computer Science Department}
}

\begin{abstract}
Automated market makers (AMMs) are automata that trade electronic assets at rates set by mathematical formulas. AMMs are usually implemented by smart contracts on blockchains.
In practice,
AMMs are often \emph{composed}:
trades can be split across AMMs,
and outputs from one AMM can be directed to another.
This paper proposes a mathematical model for AMM composition.
We define sequential and parallel composition operators for AMMs
in a way that ensures that AMMs are closed under composition,
in a way that works for ``higher-dimensional'' AMMs that
manage more than two asset classes, and so the composition
of AMMs in ``stable'' states remains stable.

\end{abstract}


\maketitle

\section{Introduction}
Decentralized finance (or ``DeFi'') has become a booming area of
distributed computing.
For example, between June 2020 and October 2020,
the total value of assets locked in decentralized finance (DeFi) protocols
surged from \$1 billion to \$7.7 billion~\cite{defipulse}.
An \emph{automated market maker} (``AMM'') is an automaton that has
custody of several pools of assets,
sets prices for those assets according to a mathematical formula,
and is always willing to trade those assets at those prices.
Unlike traditional ``order book'' traders,
AMMs do not need to match up (and wait for) compatible buyers and sellers.
Today, AMMs such as Uniswap~\cite{uniswap}, Bancor~\cite{bancor}, Balancer~\cite{balancer},
and others~\cite{pourpouneh} have become one of the most popular ways to trade electronic assets.

Here is an example of a \emph{constant-product} AMM
loosely based on Uniswap v1~\cite{uniswap}.
The AMM has state $(x,y)$ if it has custody of $x$ units of asset $X$,
and $y$ units of asset $Y$.
The AMM's state is subject to the invariant that the product $xy$ is
constant.
The AMM's states thus lie on the hyperbolic curve $x y = c$, 
for $x,y > 0$ and constant $c > 0$.
If a trader transfers $dx$ units of $X$ to the AMM, 
the AMM will transfer $dy$ units of $Y$ back to the trader,
preserving the invariant $(x+dx) (y - dy) = x y = c$.
The client profits if the value of $dy$ units of $Y$ exceeds
the value of $dx$ units of $X$ in the current market (or at another AMM).

The \emph{price} of asset $Y$ in units of $X$ at state $(x,y)$
is the curve's slope at that point.
Trades move the AMM's state along the curve:
trading $X$ for $Y$ makes $X$ cheaper and $Y$ more expensive,
a phenomenon known as \emph{slippage}.
Usually, an AMM's state reflects current market conditions:
if a constant-product AMM is in state $(x,y)$,
then the market value of $x$ units of $X$ should be the same as $y$ units of $Y$,
because otherwise a trader can make an arbitrage profit by buying
the undervalued asset.
Note that a constant-product AMM can adjust to any
(finite, non-zero)
market rate between $X$ and $Y$,
and every AMM state matches some market valuation.

It is natural to compose AMMs.
If AMM $A$ trades assets $X$ and $Y$,
and $B$ trades assets $Y$ and $Z$,
then a trader can buy $Z$ with $X$ by transferring $X$ to $A$,
feeding $A$'s $Y$ output to $B$, and pocketing $B$'s $Z$ output.
Existing systems of AMMs encourage exactly this kind of composition:
for Uniswap v1, the intermediate asset $Y$ would be ether (ETH),
and for Bancor, it would be Bancor network token (BNT).
Here is the question at the heart of this paper:
can we treat the result of composing these two constant-product AMMs
as a ``black box'' AMM for trading $X$ for $Z$?
Note that this composition is not itself a constant-product AMM,
so the class of constant-product AMMs is not closed under
composition for any reasonable notion of composition.
Can we instead pick a broader AMM definition that does support
common-sense notions of composition?
What constraints on pricing formulas yield AMMs that behave reasonably under composition?
What should composition mean when AMMs trade more than two kinds of assets?
In short, what notions of AMM composition make sense?

This paper makes the following contributions.
\begin{itemize}
\item 
We give an axiomatic characterization for well-behaved AMMs.
These axioms build on prior work,
but require small but critical changes to support
mathematical properties such as composition.

\item
For AMMs satisfying these axioms,
we show there is a duality between asset valuations,
and AMM states where assets are balanced to reflect those valuations.
Computation can be done in whichever domain is more convenient.

\item
The paper's principal contribution is to
propose novel operators for sequential and parallel
composition of AMMs.
Well-structured notions of composition for
``higher-dimensional'' AMMs that manage more than two asset classes
requires novel intermediate projection and virtualization operators.
\end{itemize}
Properly defined,
AMMs are mathematical objects that are closed under sequential and parallel composition.
We hope this paper will encourage further research into how AMMs can be combined
into networks and what such networks can do.

This paper is organized as follows.
\secref{model} describes the problem and the model of computation,
\secref{definitions} presents in axiomatic form the properties a practical AMM should satisfy,
\secref{topology} shows that all such AMMs have a common underlying topological structure,
and \secref{operators} presents basic mathematical operators useful for defining composition.
\secref{sequential} defines \emph{sequential composition},
where the outputs of one trade becomes the inputs to another.
Some careful choices are needed to impose order on composition of ``higher-dimensional''
AMMs that manage more than two assets.
\secref{parallel} defines \emph{parallel composition},
where a trader decides how to split a trade among alternative AMMs.
\secref{fees} shows how AMMs with fees fit into out composition framework,
\secref{related} surveys related work,
and \secref{conclusions} discusses future directions and open problems.

\section{Model}
\seclabel{model}
We use the following notation and terminology.
Vectors are in bold face $(\bx)$ and scalars in italics ($x$).
Variables, scalar or vector,
are usually taken from the end of the alphabet ($x,y,z$),
and constants from the beginning ($a,b,c$).
If $\ba = (a_1, \ldots, a_n)$ and
$\bb = (b_1, \ldots, b_m)$,
then $(\ba,\bb)$ is the vector $(a_1, \ldots, a_n, b_1, \ldots, b_m)$.
Vector comparisons are component-wise:
for vectors of the same dimension,
$\ba \leq \bb$ means that each $a_i \leq b_i$.
Constant vectors
$\bone = (1, \ldots, 1)$ and $\bzero = (0, \ldots, 0)$ have dimension that will be clear from context.
$\PosReals = \set{x \in \Reals | x > 0}$,
and $\NegReals = \set{x \in \Reals | x < 0}$.
We use $\bx \lneqq \bx'$ to mean that $\bx \leq \bx'$ but $\bx \neq \bx'$,
so at least one inequality is strict.

A function $f: \Reals^n \to \Reals$ is \emph{strictly convex}
if for all $t \in (0,1)$ and distinct $\bx,\bx' \in \Reals^n$,
$f(t \bx + (1-t) \bx') < tf(\bx) + (1-t)f(\bx')$.
A convex function's tangent line or plane lies below the function's curve or surface.
The function is \emph{strictly concave} if the inequality is reversed.
A set $X$ is \emph{strictly convex} if for all distinct
$x,x' \in X$ and $t \in (0,1)$,
$t x + (1-t) x'$ lies in the interior of $X$.

For $A: \PosReals^n \to \Reals$,
the set $\set{\bx \in \PosReals^n | A(\bx) = c}$ is called a \emph{level set} at $c$,
and the set $\set{\bx \in \PosReals^n | A(\bx) \geq c}$
is called the \emph{upper contour set} at $c$.

We use L1 and L2 norms: if $\bv = (v_1, \ldots, v_n)$ is a vector,
then $\|\bv\|_1 = \sum_{i=1}^n |v_i|$ and 
$\|\bv\|_2 = \sqrt{\sum_{i=1}^n v_i^2}$.

\subsection{AMM State Spaces are Manifolds}
An \emph{asset} might be a cryptocurrency, a token,
an electronic deed to property, and so on.
Assets can fluctuate in value,
and participants may want to trade assets,
perhapsto respond to past price changes,
or to anticipate future price changes.

An \emph{$n$-dimensional} AMM trades across assets $X_1, \ldots, X_n$.
Each AMM state has the form $\bx = (x_1,\ldots,x_n) \in \PosReals^n$,
where each $x_i$ is the amount of units of asset $X_i$ in the AMM's custody.
AMM states are points in a (twice) differentiable manifold,
a higher-dimensional generalization of a surface or curve.
A \emph{trade} moves the AMM from state
$\bx = (x_1,\ldots,x_n)$ to state $\bx' = (x_1',\ldots,x_n')$.
For each $i$ where $x_i' > x_i$ ,
the trader pays $x_i'-x_i$ units of $X_i$ to the AMM,
while if $x_i' < x_i$,
the AMM pays $x_i-x_i'$ units of $X_i$ to the trader.
We call $\bx-\bx'$ a \emph{profit-loss vector}.
Requiring the AMM state space to be a differentiable manifold ensures that
both prices and slippage change gradually rather than abruptly,
although we will see in \secref{definitions} that not every manifold makes sense as an AMM state space.

AMMs typically charge per-transaction fees.
For example, Uniswap v1 charges a 0.3\% fee on trades,
and the sums collected are added to the AMM's assets.
For ease of exposition, we focus initially on AMMs that do not divert
fees to their own liquidity pools
(instead, fees might be paid to a separate account).
In \secref{fees},
we show how an AMM with a Uniswap-style fee structure can be modeled as
the sequential composition of a no-fee AMM with a specialized ``linear'' AMM.

\subsection{System Model}
There are two kinds of participants in decentralized finance.
(1) \emph{Traders} transfer assets to AMMs, and receive assets back.
Traders can compose AMMs into networks in complicated ways.
(2) \emph{Liquidity providers} (or ``providers'') fund the AMMs by
lending assets, and receiving shares, fees, or other profits.
Traders and providers play a kind of alternating game:
traders modify AMM states by adding and removing assets,
and providers can respond by adding or removing assets,
reinvesting fees, or adjusting other AMM properties.

Today, AMMs are usually implemented by smart contracts on blockchains.
For our purposes,
a \emph{blockchain} is a highly-available, tamper-proof distributed ledger
that records which participants own various assets.
A \emph{smart contract} is a public program that controls how assets are
recorded and transferred on a blockchain.
Smart contracts typically support \emph{atomic transactions} that
allow traders to execute atomic sequences of trades on multiple AMMs.
The analysis given in this paper is largely independent of the
the particular technology used to implement blockchains and smart contracts.

\section{Definitions}
\seclabel{definitions}
\subsection{Common-Sense Axioms}
\seclabel{commonsense}
Although any automaton that trades assets can be regarded as an AMM,
only those AMMs that satisfy certain properties
make sense in practice.
To make this presentation self-contained,
we list some informal, common-sense axioms any practical AMM should satisfy,
and then restate these informal requirements as more precise mathematical properties.
These properties mirror prior proposals~\cite{AngerisC2020,Berenzon2020,krishnamachari2021dynamic},
with some adjustments to encompass higher-dimensional AMMs 
(those that trade more than two kinds of assets),
and to facilitate later introduction of composition operators.

\begin{informal}[Continuity]
Every AMM state should define precise rates of exchange
between every pair of assets,
and trades should change these rates gradually rather than abruptly.
\end{informal}

\begin{informal}[Expressivity]
An AMM must be able to adapt to any market conditions.
\end{informal}
If an AMM is unable to adapt to market conditions that
cause one asset to be undervalued with respect to the others,
then traders will drain all of the undervalued asset from the AMM
at the expense of the providers.

\begin{informal}[Stability]
Every AMM state should be the appropriate response to some possible market condition.
\end{informal}
If no market condition justifies entering a particular state,
then that state is superfluous.

\begin{informal}[Convexity]
Slippage should work to the disadvantage of the trader.
Buying more of asset $X$ should make $X$ more expensive, not less.
\end{informal}

Otherwise, a runaway effect can occur where traders are motivated
to buy more and more of an asset until the AMM's supply is exhausted.

The constant product AMM $A:=(x,c/x)$ is an example of an AMM
that conforms to these axioms.
By contrast,
consider the \emph{constant-sum} AMM $C:=a x + b y = c$,
which trades between assets $X$ and $Y$ at a fixed exchange rate.
Constant-sum AMMs fail to satisfy expressivity:
as long as the exchange rate matches the market rate,
a constant-sum AMM trades without slippage,
but as soon as the market rate departs from the AMM's exchange rate,
arbitrage traders will exhaust the AMM's supply of the undervalued asset,
to the detriment of the liquidity providers.
For this reason, with very few exceptions~\cite{Mstable},
constant-sum AMMs are not used in practice.
(Note, however, that the continuity axiom implies
that any AMM's behavior \emph{approximates} the behavior of
a constant-sum market maker when trades are sufficiently small.)
Henceforth, except when explicitly noted,
we use ``AMM'' as shorthand for ``AMM that satisfies these axioms''.

\subsection{AMMs and Valuations}
Formally, an $n$-dimensional AMM is given by a function 
$A:\PosReals^n \to \Reals$ such that:
\begin{itemize}
\item
  For all $c \geq 0$,
  the \emph{upper contour set} $A(\bx) \geq c$ is closed and strictly convex.
\item
  $A(\bx)$ is strictly increasing in each coordinate, and
\item
  $A$ is twice-differentiable.
\end{itemize}
We adopt the convention that the AMM's state space is the level set $A(\bx) = 0$, though sometimes we replace 0 with another constant.
For brevity, when there is no danger of confusion,
we use $A$ to refer to the AMM's function,
its state space, and the AMM itself.
We use $\upper(A)$ for $A$'s upper contour set at $0$.
We sometimes define an AMM by saying $A := F(x) = 0$
to mean the AMM $A$'s states lie on the curve $F(x) = 0$

We remark that restricting the domain of an AMM's function to all-positive coordinates is a ``without loss of generality'' convention,
since an AMM's trading behavior is unaffected
by any linear change of variables.

It is often convenient to express an AMM in an alternative form.
The implicit function theorem ~\cite{implicit} implies that for any point on the manifold,
all but one coordinate can be chosen freely,
and the remaining coordinate is a twice-differentiable 
convex function $f_i$ of the rest:
\begin{multline*}
  A(x_1, \ldots, x_{i-1}, f_i(x_1, \ldots, x_{i-1}, x_{i+1}, \ldots x_n), x_{i+1}, \ldots x_n) = 0.
\end{multline*}

An opinion on the relative values of assets $X_1, \ldots, X_n$ is
captured by a \emph{valuation} $\bv = (v_1, \ldots, v_n)$,
where\footnote{For simplicity, we assume valuations never assign an asset relative value 0 or 1.} $0 < v_i < 1$ and $\sum_i v_i = 1$.
A trader who moves an AMM from state $\bx$ to state $\bx'$
makes a profit if the dot product $\bv \cdot (\bx - \bx')$ is positive,
and otherwise incurs a loss.
The current \emph{market value} is a valuation accepted by most participants.

The \emph{standard simplex} $\Delta^n \in \PosReals^n$ is the set of points
$(v_1,\ldots,v_n)$ where for $i=1,\ldots,n$, $0 < v_i < 1$, and $\sum^n_{i=1} v_i = 1$.
Each valuation forms the barycentric coordinates of a point
in the interior of the standard $n$-simplex\footnote{This notation is non-standard but convenient; others define $\Delta^n$ to be a subset of $\PosReals^{n+1}$.}: $\int(\Delta^n)$.

A \emph{stable point} for an AMM $A$ and valuation $\bv$
is a point $\bx \in A$ that minimizes the dot product $\bv \cdot \bx$.
If $\bv$ is the market valuation,
then any trader can make an arbitrage profit
by moving the AMM from any state to a stable point,
and no trader can make a profit by moving the AMM out of a stable point.

Of course, this model is idealized in several ways.
Asset pools are not continuous variables: they assume discrete values.
Computation is not infinite-precision:
round-off errors and numerical instability are concerns.
Popular AMM such as Uniswap v2~\cite{uniswapv2} and v3~\cite{uniswapv3} perform trades under a more complicated and dynamic model than the one
considered here.
Nevertheless, the problem of AMM composition remains largely unaddressed,
at least in a formal way,
and we believe our model and definitions capture enough of the
essential properties of AMMs to make useful progress.

We assume that the functions defining an AMM do not change over time.
In practice, an AMM's defining function can change over time.
For example, AMMs charge fees in a variety of ways,
usually adding those fees to the assets managed by the AMM.
Nevertheless, an AMM's defining function does not
change in the course of a single transaction,
the duration for which AMM composition is meaningful.
We will further discuss the effects of fees in \secref{fees}.

\subsection{Formal Axioms}
We are now able to restate the common-sense axioms of
\secref{commonsense} in more precise terms.

\begin{axiom}[Continuity]
For every AMM $A$, the function $A:\PosReals^n \to \mathbb{R}$
is twice-differentiable.
\end{axiom}
As far as we know,
existing AMMs use smooth (infinitely differentiable) functions.

\begin{axiom}[Convexity]
For every AMM $A$, $\upper(A)$ is strictly convex.
\end{axiom}

The following is a standard result from convex analysis.

\begin{lemma}
    \lemmalabel{epigraph}
    A function $f: \PosReals^n \to \PosReals$ is a strictly convex function if and only if its epigraph \\
    $\epi(f) = \set{(x,a) \in \PosReals^n \times \PosReals: f(x) \leq a}$ is a strictly convex set.
\end{lemma}

\begin{axiom}[Expressivity]
Every valuation has a unique stable point in $A$.
\end{axiom}

\begin{lemma}
  \lemmalabel{stable-upper}
For any AMM $A$,
every valuation $\bv \in \int(\Delta^n)$ has a stable point in $\upper(A)$.
\end{lemma}

\begin{proof}
    Pick any $\bx$ in $\upper(A)$.
    The set 
    \begin{equation*}
        S = \set{\bx^{'} \in \NNegReals^n:\bv \cdot \bx^{'} \leq \bv \cdot \bx}
    \end{equation*}
    is compact.
    Since $\upper(A)$ is closed, $\tilde{S} = S \cap \upper(A)$ is compact.
    A stable point solves the optimization problem
    \begin{equation*}
        \min_{\bx^{'} \in \tilde{S}} \bv \cdot \bx^{'}.
    \end{equation*}
    This minimum exists since $\bv \cdot \bx^{'}$
    is a continuous function on the compact set $\tilde{S}$.
\end{proof}

\begin{lemma}
  \lemmalabel{stable-unique}
For any AMM $A$ and any valuation $\bv$,
the stable point for $\bv$ in $\upper(A)$ is unique.
\end{lemma}

\begin{proof}
  Fix valuation $\bv \in \int(\Delta^{n})$ and
  let $\bx,\bx^{'} \in \upper(A)$ where $\bx \neq \bx^{'}$
  and $w = \bv \cdot \bx = \bv \cdot \bx^{'}$.
  Because $\upper(A)$ is strictly convex,
  for all $t \in (0,1)$,
  $\tilde{\bx} = t \bx + (1 - t)\bx^{'}$ is in $\int(\upper(A))$.
  Since $\int(\upper(A))$ is open,
  there is some $\epsilon > 0$ such that the open
  $\epsilon$-ball $B_{\tilde{\bx}}(\epsilon) \subset \int(\upper(A))$.
  Now choose $\bx^{*}$ in $B_{\tilde{\bx}}(\epsilon)$
  such that $\bx^{*} < \tilde{\bx}$.
  Then we have
  \begin{align*}
    \bv \cdot \bx^{*}
    &< \bv \cdot \tilde{\bx}\\
    &= t \bv \cdot \bx + (1 - t)\bv \cdot \bx^{'}\\
    &= t w + (1-t) w = w,
\end{align*}
  a contradiction.
\end{proof}

For example,
for the 2-dimensional constant-product AMM given by $(x,c/x)$,
the valuation $(v,1-v)$ has the unique stable point
\begin{equation*}
    \left(\sqrt{\frac{c(1-v)}{v}},\sqrt{\frac{v}{c(1-v)}}\right).
\end{equation*}
More generally,
for the 2-dimensional constant-product AMM given by $(x,f(x))$,
the valuation $(v,1-v)$ has the unique stable point
\begin{equation*}
\left(f^{\prime -1}\left(-\frac{v}{1-v}\right), f\left(f^{\prime -1}
\left(-\frac{v}{1-v}\right)\right)\right).
\end{equation*}
\begin{lemma}
  \lemmalabel{stable-bdry}
  For all valuations $\bv$,
  the stable point for $\bv$ in $\upper(A)$ lies on the level set $A$.
\end{lemma}

\begin{proof}
    Fix valuation $\bv$ and let $\bx^{*}$ be its stable point.
    Suppose that $\bx^{*} \not \in A$ but $\bx^{*} \in \int(\upper(A))$.
    As in \lemmaref{stable-unique},
    we can find $\epsilon > 0$ and an open ball $B_{\bx^{*}}(\epsilon)$ in $\int(\upper(A))$.
    Choosing $\bx \in B_{\bx^{*}}(\epsilon)$ such that $\bx< \bx^{*}$,
    we have $\bv \cdot \bx < \bv \cdot \bx^{*}$,
    a contradiction.
\end{proof}

\begin{corollary}
  Every AMM satisfies expressivity.
\end{corollary}

\begin{axiom}[Stability]
Every $\bx \in A$ is the stable point for some valuation.
\end{axiom}

For example,
for the 2-dimensional constant-product AMM $(x,c/x)$,
the point $(x,c/x)$ is the stable point for the
valuation $(\frac{c}{c+x^2},1-\frac{c}{c+x^2})$.
More generally,
for the 2-dimensional constant-product AMM $(x,f(x))$,
the point $(x,f(x))$ is the stable point for the
valuation $(\frac{f'(x)}{f'(x)-1},\frac{1}{1-f'(x)})$.

\begin{lemma}
\lemmalabel{every-y-some-w}
  For every $\bx \in A$,
  there is some $\bw \in \NegReals^n$ such that $\bw \cdot \bx > \bw \cdot \by$ for all $\by \in \upper(A)$, $\by \neq \bx$. 
\end{lemma}

\begin{proof}
    Pick $\bx \in A$. 
    Since $\upper(A)$ is strictly convex,
    the supporting hyperplane theorem implies there is a non-zero
    $\bw \in \Reals^{n}$ such that
    $\bw \cdot \bx > \bw \cdot \by$ for all $\by \in upper(A)$, $\by \neq \bx$.
    We need to show that $\bw \in \NegReals^n$.
    Say $\bw \in \PosReals^n$.
    Choose any $\epsilon > 0$ and let $\tilde{\bx} = \bx + \epsilon \bw$
    so that $\bw \cdot \tilde{\bx} = \bw \cdot \bx + \epsilon\|\bw\|^2 > \bw \cdot \bx$.
    By monotonicity, $\tilde{\bx} \in \upper(A)$, a contradiction.
    Now consider the case where $\bw \in \Reals^n \setminus (\NegReals^n \cup \PosReals^n)$,
    namely $\bw$ cannot have all strictly positive or all strictly negative entries.
    We construct $\tilde{\bw}$ orthogonal to $\bw$.
    Replace all of the non-negative entries of $\bw$ in $\tilde{\bw}$
    with the sum of the absolute values of the negative coordinates.
    Replace all the negative entries of $\bw$ in $\tilde{\bw}$
    with the sum of all of the non-negative entries of $\bw$.
    These replacements guarantee that all coordinates of $\tilde{\bw}$ are positive,
    and $\tilde{\bw} \cdot \bw = 0$.
    Pick $\epsilon > 0$ and let
    $\tilde{\bx} = \bx + \epsilon \tilde{\bw}$
    so
    \begin{align*}
      \bw \cdot \tilde{\bx} 
      &= \bw \cdot \bx + \epsilon \bw \cdot \tilde{\bw} \\
      &= \bw \cdot \bx
    \end{align*}
and yet $\tilde{\bx} \in \upper(A)$ since $A$ is strictly increasing.
By contradiction, $\bw \in \NegReals^n$.
\end{proof}

\begin{lemma}
\lemmalabel{stable-for-some}
  For every $\bx \in A$,
  there exists a valuation $\bv$ for which $\bx$ is a stable point.
\end{lemma}

\begin{proof}
    Let $A$ be an $n$-dimensional AMM and let $\bx \in A$.
    Choose $\bw \in \NegReals^n$ as described in the previous lemma for $\bx$.
    We have $\bw \cdot \bx > \bw \cdot \by$ for $\by \in A$, $\by \neq \bx$.
    Negating $\bw$ and re-scaling the result so the elements sum to $1$ yields
    $\bv \in \int(\Delta^n)$.
    Thus we have $\bv \cdot \bx < \bv \cdot \by$ for all $\by \in A$ where $\by \neq \bx$,
    implying $\bx$ is a stable point for $\bv$.
\end{proof}

\begin{corollary}
  Every AMM satisfies stability.
\end{corollary}

In short, the goal of this paper is to balance axioms and composition operators so that the class of AMMs satisfying
these axioms remains closed under these composition operators.

\section{Topological Equivalence}
\seclabel{topology}
How many truly distinct AMMs of a given dimension are there?
Considered as a mathematical object,
much of an AMM's structure is captured by the link between valuations and their stable points.
We say that two AMMs are \emph{topologically equivalent}
if there is a stable-point preserving homeomorphism between their manifolds.
By itself, a homeomorphism between manifolds conveys little information,
but a homeomorphism that preserves stable points preserves the AMMs' common underlying structure.

In this section we show that \emph{all} AMMs over the same set of assets,
if they satisfy our axioms, are topologically equivalent.
More precisely,
for any two AMMs $A(x_1,\ldots,x_n)$ and $B(x_1,\ldots,x_n)$ over asset types
$X_1, \ldots, X_n$, satisfying our axioms,
there is a homeomorphism $\mu: A \to B$
such that $\bx$ and $\mu(\bx)$ are the stable states for the same valuation.
This proof relies on the uniqueness of stable points:
for example it would not hold if AMM functions were convex
instead of strictly convex.

Although topological equivalence implies a common mathematical structure,
two topologically equivalent AMMs may differ substantially
with respect to price slippage, fees,
or how expensive it is to move from one valuation's stable state to another's.

Recall from \lemmaref{stable-unique} and \lemmaref{stable-for-some}
that there is a unique function $\phi: \int(\Delta^n) \to A$
carrying each valuation to its unique stable point.

\begin{lemma}
\lemmalabel{stable-state-biject}
\sloppy
For $A$, an AMM,
the stable point map
\begin{equation*}
    \phi:~\int(\Delta^n)~\to~A
\end{equation*}
is a continuous bijection.
\end{lemma}

\begin{proof}
The map $\phi$ is surjective by \lemmaref{stable-for-some},
and injective by \lemmaref{stable-unique}.
To show continuity,
consider the sequence
$\set{\bv_n}_{n=1}^{\infty} \subset \int(\Delta^n)$ where $\lim_{n \to \infty}\bv_n = \bv$.
Let $\tilde{\bx} = \lim_{n \to \infty} \phi(\bv_n)$ and $\bx^{*} = \phi(\bv)$.
  Suppose $\tilde{\bx} \neq \bx^{*}$.
  Note that $\bv \cdot \bx^{*} < \bv \cdot \tilde{\bx}$ by definition of stable point.
  Letting $\overline{\bx} = \frac{\bx^{*} + \tilde{\bx}}{2}$ by strict convexity we know $\overline{\bx} \in \int(\upper(A))$.
  We also have that $\bv \cdot \bx^{*} < \bv \cdot \overline{\bx} < \bv \cdot \tilde{\bx}$.
  Notice now that $\bv_n \cdot \overline{\bx} > \bv_n \cdot \phi(\bv_n)$ by definition so taking limits we get
  $\bv \cdot \overline{\bx} > \bv \cdot \tilde{\bx}$, a contradiction.
  Thus $\lim_{n \to \infty} \phi(\bv_n) = \phi(\bv)$.
\end{proof}

\begin{lemma}
\lemmalabel{stable-state-homeo}
For $A$, an AMM,
the stable point map
\begin{equation*}
    \phi:~\int(\Delta^n)~\to~A
\end{equation*}
is a homeomorphism.
\end{lemma}

\begin{proof}
From \lemmaref{stable-state-biject}, $\phi$ is both bijective and continuous,
so it is enough to show $\phi^{-1}$ is continuous.
For any $\bv$ with stable point $\bx$,
the first-order conditions imply that $\bv = \lambda \nabla A(\bx)$
for some non-zero (Lagrange multiplier) $\lambda \in \Reals$.
Since $A$ is strictly increasing, $\lambda > 0$.
Thus we can think of $\bv$ as a function of $\bx$,
written $\bv(\bx) = \lambda(\bx) \nabla A(\bx)$.
Because $A$ is continuously differentiable, $\nabla A(\bx)$ is continuous,
so it is enough to check $\lambda(\bx)$ is continuous.
Because $\bv(\bx)$ is a convex combination, and $\lambda(\bx)$ is unique,
\begin{equation*}
\lambda(\bx) = \frac{1}{\|\nabla A(\bx)\|_1},
\end{equation*}
which is continuous because each $\frac{\partial A(\bx)}{\partial x_i}$ is continuous.
It follows that $\phi^{-1}$ is continuous.
\end{proof}

\begin{theorem}
Let $A$ and $B$ be AMMs over the same set of assets.
There is a homeomorphism $\mu: A \to B$
that preserves stable points:
if $\bx$ is the stable point for valuation $\bv$ in $A$,
then $\mu(\bx)$ is the stable point for $\bv$ in $B$.
\end{theorem}

\begin{proof}
  By \lemmaref{stable-state-homeo},
  there exist homeomorphisms
  \begin{align*}
  \phi: \int(\Delta^n) &\rightarrow A,\\
  \phi': \int(\Delta^n) & \rightarrow B
  \end{align*}
  Their composition $\mu = \phi' \circ \phi^{-1}$ is also a homeomorphism.
  For $\bv \in \int(\Delta^n)$ with stable points $\bx \in A, \bx' \in B$,
  $\phi(\bv) = \bx$ and $\phi^{'}(\bv) = \bx^{'}$,
  so $\mu(\bx) = \phi^{'}(\phi^{-1}(\bx)) = \phi^{'}(\bv) = \bx^{'}$,
  implying that $\mu$ preserves stable points.
\end{proof}
Note that this result requires that $A$ have continuous first derivatives.

\section{Operators}
\seclabel{operators}
It is useful to be able to reduce an AMM's dimension,
perhaps by ignoring some assets,
or by creating ``baskets'' of distinct assets that can
be treated as a unit.
In this section we introduce two tools for reducing dimensionality:
projection, and asset virtualization.

\subsection{Projection}
An AMM may provide the ability to trade across a variety of asset types,
but traders may choose to restrict their attention to a subset,
ignoring the rest.
Perhaps the ignored assets are too volatile, or not volatile enough,
or there are regulatory barriers to owning them.

Mathematically,
the \emph{projection} operator acts on an AMM by fixing some
state coordinates to constant values and letting the rest vary.
We will show that projecting an AMM in this way yields another AMM
of lower dimension.
Informally, traders are free to ignore uninteresting assets.
\begin{definition}
  Let $\bx = (x_1, \ldots, x_n)$, $\by = (y_1, \ldots, y_m)$,
  and a constant $\ba = (a_1, \ldots, a_n)$.
  The projection of $A$ onto $\ba$ is given by
  $A_{\ba}(\by) = A(\ba, \by) = 0$
\end{definition}
\begin{lemma}
  Given an $(n+m)$-dimensional AMM $A(\bx,\by) = 0$ and $\ba \in \PosReals^n$,
  the projection $A_{\ba}(\by)$ is an $m$-dimensional AMM.
\end{lemma}
\begin{proof}
  It is enough to check that $A_{\ba}$ is twice-differentiable,
  strictly increasing, and $\upper(A)$ is strictly convex.
  Because $A(\bx,\by)$ is twice-differentiable,
  so is $A(\ba,\by) = A_{\ba}(\by)$.
  To show that $A_{\ba}$ is strictly increasing,
  let $\bx' \gneqq \bx$.
  \begin{align*}
    A_{\ba}(\bx)
    &= A(\ba,\bx)\\
    &< A(\ba,\bx')\\
    &= A_{\ba}(\bx').
  \end{align*}
  To show that $\upper(A_{\ba})$ is strictly convex,
  pick distinct $\bx$ and $\bx'$ in $\upper(A_{\ba})$.
  Namely $A(\ba,\bx) = A_{\ba}(\bx) \geq 0$ and $ A(\ba,\bx^{'}) = A_{\ba}(\bx') \geq 0$.
  For $t \in (0,1)$:
  \begin{align*}
    A_{\ba}(t \bx + (1-t) \bx')
    &= A(t \ba + (1-t) \ba, t \bx + (1-t) \bx')\\
    &= A(t(\ba,\bx) + (1-t)(\ba,\bx^{'})) > 0
  \end{align*}
  by the strict convexity of $\upper(A)$.
\end{proof}

\begin{lemma}
For index set $I$,
let $\bv = (v_i | i \in I)$ be a valuation for
a sequence of asset types $X = (X_i | i \in I)$.
For $J \subset I$,
$\bv' = (\frac{v_j}{1 - \sum_{k \not \in J}v_k} | j \in J)$
is a valuation for $X' = (X_j | j \in J) \subset X$.
\end{lemma}
We say that $X'$ \emph{inherits} $\bv'$ from valuation $\bv$ of $X$.

The next lemma states that stable points persist under projection.
\begin{lemma}
\lemmalabel{project-stable}
Let $A(\bx,\by) = 0$ be an $(n+m)$-dimensional AMM,
$\bv \in \int(\Delta^{n+m})$ a valuation on $(\bx,\by)$.
If $(\ba,\bb)$ is the stable point for $\bv$ in $A(\bx,\by)$
then $\bb$ is the stable point for the inherited valuation $\bv^{'}$.
\end{lemma}
\begin{proof}
  Suppose the stable point for $\bv^{'}$ is $\bb' \neq \bb$, namely $\bv^{'} \cdot \bb^{'} < \bv^{'} \cdot \bb$.
  Scaling both sides by $1 - \sum_{i = 1}^{n} v_i$ yields
  $(v_{n+1},\ldots,v_{n+m}) \cdot \bb' \leq (v_{n+1},\ldots,v_{n+m}) \cdot \bb$.
  Adding $(v_1, \ldots, v_n) \cdot \ba$ to both sides yields
  $v \cdot (\ba,\bb') \leq v \cdot (\ba,\bb)$,
  contradicting the assumption that $(\ba,\bb)$ is the stable point for $\bv$.
\end{proof}

\subsection{Virtualization}
\seclabel{virtualization}
It is sometimes convenient to create a ``virtual asset''
from a linear combination of assets.
Here we show that replacing a set of assets
traded by an AMM with a single virtual asset
is also an AMM.
This construction works for any linear combination,
although the most sensible combination
is usually the assets' current market valuation.

Here is a simple example of asset virtualization.
Consider an AMM that trades across three asset types,
$X,Y,Z$, defined by the constant-product formula 
\begin{equation*}
    A(x,y,z) = x y z - 8 = 0,
\end{equation*}
initialized in state $(2,2,2)$.
A trader believes that 2 units of $Y$ are always worth 1 unit of $Z$,
and that it makes sense to link them in that ratio
by creating a virtual asset $W$
worth 2/3 units of $Y$ and 1/3 unit of $Z$,
and to trade in a single denomination of $W$ instead of
individual denominations of $Y$ and $Z$.

Formally, the trader defines $W$ in terms of the valuation
$\bv = (\frac{2}{3},\frac{1}{3})$ on $Y,Z$.
The virtualized AMM $A|\bv$ is defined by
\begin{align*}
  (A|\bv)(x,w) 
  &= A(x,\frac{2w}{3},\frac{w}{3}+1)\\
  &= x \frac{2w}{3}(\frac{w}{3}+1) - 8 \\
  &= 0,
\end{align*}
with initial state $(2,3)$.
(The ``+1'' in the $Z$ co-ordinate appears
because 2 $Y$ and 2 $Z$ units are not evenly divisible
into $W$ units.)

Let $\bx = (x_1, \ldots, x_n)$ and $\by = (y_1, \ldots, y_m)$.
Let $A(\bx,\by) := A(x_1, \ldots, x_n,y_1,\ldots,y_m) = 0$ be an
$(n+m)$-dimensional AMM with initial state
$(\ba,\bb) = (a_1, \ldots, a_n, b_1, \ldots, b_m)$.
Let us create a virtual asset $Z$ from
$y_1, \ldots, y_m$,
using the valuation $\bv = (v_1, \ldots, v_m)$.

Let $c \in \PosReals$ be the largest value such that
$\bb - c \bv \geq \bzero$.
The value $c$ is the number of $Z$ assets in $\bb$,
and $\br =\bb - c \bv$ is the vector of residues
if $\bb$ is not evenly divisible into $Z$ units.
The virtualized AMM is given by
\begin{align*}
(A|\bv)(\bx,z) 
&= A(\bx, v_1 z + r_1, \ldots, v_m z + r_m) \\
&= 0,
\end{align*}
with initial state $(a_1, \ldots, a_n, c)$.

The next lemma says that in any AMM state,
it is always possible to virtualize any set of assets.
\begin{lemma}
  Let $A$ be an $(m+n)$-dimensional AMM in state $(\ba,\bb)$,
  where $\ba \in \PosReals^m$, $\bb \in \PosReals^n$,
  and valuation $\bv \in \int(\Delta^n)$.
  We claim that for any $\ba' \in \PosReals^m$,
  there is a unique $t \in \Reals$ such that $(\ba',\bb + t \bv)$
  is a state of $A$.
\end{lemma}

\begin{proof}
We seek $t \in \Reals$ such that $A(\ba',\bb + t \bv) = 0$.
There are several cases.
If $A(\ba',\bb) = 0$, then $t=0$ and we are done.
Suppose $A(\ba',\bb) < 0$.
If $A(\ba',\bb + \bv) = 0$, then $t=1$ and we are done.
If $A(\ba',\bb + \bv) < 0$,
pick a vector $\bc = (c_1, \ldots, c_n) \in \PosReals^n$
such that $A(\ba',\bb+\bc) = 0$.
Let $\epsilon > 0$,
\begin{equation*}
s_i = \frac{c_i + \epsilon}{v_i}, 0 \leq i \leq n,
\end{equation*}
and $s = \max_{0 \leq i \leq n}s_i$.
It follows that $s \bv \geq \bc + \epsilon \bone$,
and $\bb + s \bv \geq \bb + \bc + \epsilon \bone$.
Since $A_{\ba^{'}}$ is strictly increasing,
$A(\ba',\bb + s \bv) \geq A_{\ba'}(\bb+\bc) = 0$.
Define $\alpha(t): [0,1] \to \Reals$ by $\alpha(t) = A_{\ba'}(\bb + \bv + t(s -1)\bv)$.
Because $\alpha$ is continuous,
the intermediate value theorem guarantees a unique
$t^{*} \in (0,1)$ such that $\alpha(t^{*}) = 0$.
Taking $t = (1 + t^{*}(s-1))$ establishes the claim.
If $A(\ba',\bb + \bv) > 0$,
let
\begin{equation*}
s_i = \frac{c_i - \epsilon}{v_i}, 0 \leq i \leq n,
\end{equation*}
and $s = \min_{0 \leq i \leq n} s_i$,
and the claim follows from a symmetric argument.

Suppose $A(\ba',\bb) > 0$.
Pick a vector $\bc = (c_1, \ldots, c_n) \in \PosReals^n$
such that $A(\ba',\bb-\bc) = 0$.
Let $\epsilon > 0$,
$s_i = \frac{c_i + \epsilon}{v_i}, 0 \leq i \leq n$.
and $s = \max_{0 \leq i \leq n}s_i$.
It follows that $s \bv \geq \bc + \epsilon \bone$,
and $\bb - s \bv \leq \bb - \bc - \epsilon \bone$,
so $A(\ba',\bb - s \bv) \leq A(\ba',\bb-\bc) = 0$.
Let $\alpha(t): [0,1] \to \Reals$ be
$\alpha(t) = A(\ba',\bb + (t-1)s\bv)$.
As before, the intermediate value theorem guarantees a unique
$t^{*} \in (0,1)$ such that $\alpha(t^{*}) = 0$.
Taking $t = (1-t^{*})s$ establishes the claim.
\end{proof}

\begin{theorem}
\thmlabel{virtual-amm}
Given an $(n+m)$-dimensional AMM $A(\bx,\by) = 0$,
and a valuation $\bv \in \int(\Delta^m)$,
the virtualized $(A|\bv)(\bx,z)$ is an $(n+1)$-dimensional AMM.
\end{theorem}

\begin{proof}\sloppy
  It is enough to check that $A|\bv$ is twice-differentiable,
  strictly increasing, and $\upper(A|\bv)$ is strictly convex.
  $(A|\bv)$ is twice-differentiable because $A$ is twice-differentiable.

  To show that $A|\bv$ is strictly increasing,
  let $\bx' \geq \bx$ and $z' \gneqq z$.
  \begin{equation*}
    \begin{aligned}
    (A|\bv)&(x_1,\ldots,x_n,z)\\
    &= A(x_1,\ldots,x_n, v_1 z + r_1, \ldots, v_m z + r_m)\\
    &< A(x_1',\ldots,x_n', v_1 z' + r_1, \ldots, v_m z' + r_m)\\
    &= (A|\bv)(x_1',\ldots,x_n',z').
    \end{aligned}
  \end{equation*}
  
  To show that $\upper(A|\bv)$ is strictly convex,
  pick distinct $(\bx,z)$ and $(\bx',z')$ on the manifold:
  \begin{equation*}
    (A|\bv)(\bx,z) = (A|\bv)(\bx',z') = 0.
  \end{equation*}
  For $t \in (0,1)$,
  \begin{equation*}
    \begin{aligned}
    (A|\bv)&(t \bx + (1-t) \bx', t z + (1-t)z')\\
    &= A(t \bx + (1-t) \bx',\bv (tz + (1-t)z') + \br) \\
    &=  A(t \bx + (1-t) \bx',t(\bv z + r)  + (1-t)(\bv z' + r)) \\
    &= A(t (\bx,\bv z+ r) + (1-t) (\bx',\bv z' + r)) > 0
    \end{aligned}
  \end{equation*}
  by the strict convexity of $\upper(A)$.
\end{proof}

Stable points are well-behaved under virtualization.
Let $A(\bx,\by) = 0$ be an $(n+m)$-dimensional AMM
and $\bv \in \int(\Delta^n)$ a valuation.
If $A$ is an $(n+m)$-dimensional AMM in state $(\ba,\bb)$,
and $\bv \in \int(\Delta^n)$ a valuation,
then $(A|\bv)(\ba,c) = A(\ba,c\bv + \br) = A(\ba,\bb) = 0$.
For any state in the virtualized AMM,
$(A|\bv)(\bx,t) = A(\bx,\bb + (t-c) \bv) = 0$.
This expression depends on $\bb$ and $\bv$,
where $c$ is a constant determined by $\bb$ and $\bv$.
Since $(A|\bv)(\bx,t)$ is an $(n+1)$-dimensional AMM,
we can write $t = f(\bx)$ for some $f: \PosReals^n \to \Reals$.
The virtualized AMM can be expressed as $(\bx,f(\bx))$
where $A(\bx,\bb + (f(\bx) - c)\bv) = 0$.

\begin{lemma}
\lemmalabel{virtual-stable}
\sloppy
    If $(\ba^{*},\bb^{*})$ is the stable point on AMM $A(\bx,\by)$ for valuation $(\bv,\bw)$,
    then $(\ba^{*},f(\ba^{*}))$ is the stable point on the virtualized AMM $(A| \bw)(\bx,t)$
    for the valuation $(\bv,\|\bw\|^2_2) / \| (\bv,\|\bw\|^2_2) \|_1$.
\end{lemma}

\begin{proof}
\sloppy
    Suppose $(\ba^{*},f(\ba^{*}))$ is not a stable point for $(\bv,\|\bw\|^2_2)$:
    there is a distinct point $(\ba,f(\ba)) \in A|\bw$ where $\bv \cdot \ba + \|\bw\|^2_2 f(\ba) < \bv \cdot \ba^{*} + \|\bw\|^2_2 f(\ba^{*})$.
    Now define $\bb = \bb^{*} + \bw(f(\ba) - f(\ba^{*}))$, which by the virtualization construction we have $(\ba,\bb) \in A$.
    Then
    \begin{align*}
        &\bv \cdot \ba + \bw \cdot \bb = \bv \cdot \ba + \bw \cdot \bb^{*} + \bw \cdot \bw (f(\ba) - f(\ba^{*})) \\
        &= \bv \cdot \ba + \|\bw\|^2_2 f(\ba) - \|\bw\|^2_2 f(\ba^{*}) + \bw \cdot \bb^{*} \\
        &<  \bv \cdot \ba^{*} + \|\bw\|^2_2 f(\ba^{*}) - \|\bw\|^2_2 f(\ba^{*}) + \bw \cdot \bb^{*}\\
        &= \bv \cdot \ba^{*}  + \bw \cdot \bb^{*}
    \end{align*}
    This is a contradiction since $(\ba^{*},\bb^{*})$ is the stable point for $(\bv,\bw)$.
\end{proof}

\section{Sequential Composition}
\seclabel{sequential}
AMMs are intended to be composed.
A Uniswap v1 AMM typically converts between an ERC-20 token and ether cryptocurrency.
To convert, say, florin tokens to guilder tokens,
one would first convert florins to ether, then ether to guilders.
Bancor uses a proprietary BNT token for the same purpose.
Some form of composition seems to be essential to making AMMs useful,
but we will see that while there are many ways in which AMMs might be composed,
not all of them make sense.
The most basic property one would demand is \emph{closure} under composition:
the result of composing two AMMs should itself be an AMM.

Being closed under composition should not be taken for granted.
For example, consider two constant-product AMMs:
$A := (x, 1 / x)$, initialized in state $(a,1/a)$,
and
$B := (y, 1 / y)$, initialized in state $(b,1/b)$.
Their composition $A \otimes B = (x,h(x))$, where
\begin{equation}
  h(x) = \frac{1}{b + \frac{1}{a} - \frac{1}{x}}
       = \frac{a x}{x - a + a b x}.
\end{equation}
The set of constant-product AMMs is thus not closed under composition.

\subsection{One-to-One Composition}
We first consider the result of composing 2-dimensional AMMs,
that is, AMMs that trade between two asset types.

Consider AMMs $A := (x,f(x))$, initialized to $(a,f(a))$,
and $B := (y,g(y))$, initialized to $(b,g(b))$.
$A$ trades between asset types $X$ and $Y$,
and $B$ between $Y$ and $Z$.
Their composition, initialized to $(a,g(b))$,
trades between $X$ and $Z$.

Operationally, composition is defined as follows.
\begin{itemize}
\item Move $A$ from state $(a,f(a))$ to state $(x,f(x))$,
yielding profit-loss vector $(a-x,f(a)-f(x))$.
\item Add $f(a)-f(x)$ to the $Y$ balance of $B$,
  yielding new state $(b+f(a)-f(x), g(b+f(a)-f(x)))$.
\end{itemize}
This trade
takes the composition from $(a,g(b))$ to $(x, g(b+f(a)-f(x)))$.
Let $h(x) = g(b+f(a)-f(x))$.
The composition $A \otimes B$ is given in the form $(x,h(x))$. Because $f,g$ are twice-differentiable:
\begin{lemma}
  $(A \otimes B)(x,y)$ is twice-differentiable.
\end{lemma}
\begin{lemma}
  $(A \otimes B)(x,y)$ is strictly increasing.
\end{lemma}
\begin{proof}
  We show that if $(x',y') \gneqq (x,y)$, meaning at least one coordinate is strictly greater,
  then $(A \otimes B)(x',y') > (A \otimes B)(x,y)$.
  Recall that $f$ and $g$ are strictly decreasing by hypothesis.
  There are two cases.
  First, suppose $x' > x$ and $y' \geq y$.
\begin{align*}
  x' &> x\\
  f(x') &< f(x)\\
  b+f(a)-f(x') &> b+f(a)-f(x)\\
  g(b+f(a)-f(x')) &< g(b+f(a)-f(x))\\
  y-g(b+f(a)-f(x')) &> y-g(b+f(a)-f(x))\\
  y'-g(b+f(a)-f(x')) &> y-g(b+f(a)-f(x))\\
  (A \otimes B)(x,y') &> (A \otimes B)(x,y)
\end{align*}
  The second case, where $x' \geq x$ and $y' > y$ is similar.
\end{proof}
\begin{lemma}
  The upper contour set $\upper(A \otimes B)$ is strictly convex.
\end{lemma}
\begin{proof}
  Since $\upper(A \otimes B) = \epi(h)$, by \lemmaref{epigraph}, it is enough to check the strict convexity of $h$.
  Pick two distinct $x$ and $x'$.
  Recall that for $t \in (0,1), f(t x + (1-t)x') < t f(x) + (1-t)f(x')$,
  and similarly for $g(y)$ by hypothesis.
  For $t \in (0,1)$,
  \begin{equation*}
    f((1-t)x + t x') \\
    < (1-t)f(x) + t f(x')
  \end{equation*}
  \begin{multline*}
    b+f(a)-f((1-t)x + t x') \\
    > (1-t)(b+f(a)-f(x))\\ + t(b+f(a)-f(x'))
  \end{multline*}
  \begin{multline*}
      g(b+f(a)-f((1-t)x + t x')) \\
      < (1-t)g(b+f(ag)-f(x)) \\+ tg(b+f(a)-f(x')))
  \end{multline*}
  \begin{equation*}
    h((1-t)x + t x) < (1-t)h(x) + t h(x')\\
  \end{equation*}
  which proves the claim.
\end{proof}
The next lemma relates stability and sequential composition
for 2-dimensional AMMs.

\begin{lemma}
  Let $(v,w,v')$ be a valuation for assets $X,Y,Z$.
  If $(a,f(a))$ is the stable point for AMM $A := (x,f(x))$
  for valuation $(v,w)/ \|(v,w)\|_1$
  and $(b,g(b))$ the stable point for AMM $(y,g(y))$
  for valuation $(w,v')/ \|(w,v')\|_1$,
  then $(a,h(a))$ is the stable point for the valuation $(v,v')/\|(v,v')\|_1$.
\end{lemma}

\begin{proof}
If $(a,h(a))$ is not the stable point for $(v,v')$,
there is some $x \neq a$ such that $(x,h(x))$ is the stable point.
By assumption, for $x \neq a$ and $y \neq b$,
\begin{align*}
    &v a + wf(a) < v x + w f(x)\\
    &wb + v'g(b) < wy + v'g(y)
\end{align*}
Also by assumption,
\begin{align*}
  vx + v'h(x)
  &< va + v'h(a)\\
  &= va +v'g(b + f(a) - f(a))\\
  &= va + v'g(b) \\
  &= va + wf(a) - wf(a) + wb - wb + v'g(b) \\
    &< vx + wf(x) - wf(a) + wy + v'g(y) - w b\\
    &= vx + w(f(x) - f(a)) + w(y - b) + v'g(y) \\
    &= vx + w(f(x) - f(a)) + w(f(a) - f(x))\\
    &\quad \quad + v'g(b + f(a) - f(x))\\
    &= vx +v'h(x),
\end{align*}
a contradiction.
(The last step follows by taking $y = b + f(a) - f(x) \neq b$,
possible since $f$ strictly decreasing).
\end{proof}

The converse is false.
Consider two constant product AMMs $A := (x,\frac{1}{x})$
and $B := (y,\frac{1}{y})$,
both initially in state $(1,1)$.
The composed AMM is given by $A \otimes B$ $(x,\frac{x}{2x-1})$ with state $(1,1)$.
Now if $(v,w,v') = (\frac{1}{4},\frac{1}{2},\frac{1}{4})$,
$(1,1)$ is the stable point of $A \otimes B$ with respect to $(v,w,v')$.
However, the stable point for $(v,w)$ on $A$ is
$(\sqrt{2},\frac{\sqrt{2}}{2})$
and $(w,v')$ on $B$ is $(\frac{\sqrt{2}}{2},\sqrt{2})$.

\subsection{Many-to-One Composition}
\seclabel{many-to-one}
AMM $A$ trades asset types $X_1, \ldots, X_m, Z$,
with initial state $(\ba, f(\ba))$,
where $\ba=(a_1,\ldots,a_m)$,
and $f(\ba)$ is the explicit function
defining the $Z$ coordinate in terms of the others.

AMM $B$ trades asset types $Z, Y_1, \ldots, Y_n$,
with initial state $(c,\bb, g(c,\bb))$,
where $\bb=(b_1,\ldots,b_{n-1})$,
and $g(c,\bb)$ is the function
defining the $Y_n$ coordinate in terms of the others.

The $Z$ asset flows between $A$ and $B$ but is
not directly accessible to traders.
The composition $A \otimes B$ trades asset types
$X_1, \ldots, X_m, Y_1, \ldots, Y_n$,
with initial state $(\ba, \bb, h(\ba,\bb))$,
for $h$ to be defined.

Operationally, the composition works as follows.
The trader changes each $a_i$ to $x_i$, $0 \leq i \leq m$,
and each $b_i$ to $y_i$, $0 \leq i \leq n-1$.
Let $\bx=(x_1,\ldots,x_m)$ and $\by=(y_1,\ldots,y_{n-1})$.
The new state of $A$ is $(\bx,f(\bx))$.
The amount of $Z$ that flows from $A$ to $B$ is $f(\ba)-f(\bx)$.
The new state of $B$ is
$(c+f(\ba)-f(\bx), \by, g(c+f(\ba)-f(\bx)))$.
The new state of $A \otimes B$ is
$(\bx,\by,g(c+f(\ba)-f(\bx)))$

Define $h(\bx,\by) = g(c+f(\ba)-f(\bx), \by)$.
$(A \otimes B)(\bx,\by,z) = z - h(\bx,\by)$.
\begin{lemma}
  $(A \otimes B)(\bx,\by,z)$ is twice-differentiable.
\end{lemma}
\begin{proof}
  Immediate because $f,g$ are twice-differentiable by hypothesis.
\end{proof}
\begin{lemma}
  $(A \otimes B)(\bx,\by,z)$ is strictly increasing.
\end{lemma}
\begin{proof}
  We show that if $(\bx',\by',z') \gneqq (\bx,\by,z)$,
  then $(A \otimes B)(\bx',\by',z') > (A \otimes B)(\bx,\by,z)$.
  There are two cases.
  First, suppose $(\bx',\by') \gneqq (\bx,\by)$ and $z' \geq z$.
\begin{align*}
  (\bx',\by') &\gneqq (\bx,\by)\\
  f(\bx',\by') &< f(\bx,\by)\\
  \bb+f(\ba)-f(\bx',\by') &> \bb+f(\ba)-f(\bx,\by)\\
  g(\bb+f(\ba)-f(\bx',\by')) &< g(\bb+f(\ba)-f(\bx,\by))\\
  z-g(\bb+f(\ba)-f(\bx',\by')) &> z-g(\bb+f(\ba)-f(\bx,\by))\\
  z'-g(\bb+f(\ba)-f(\bx',\by')) &> z-g(\bb+f(\ba)-f(\bx,\by))\\
  (A \otimes B)(\bx',\by',z') &> (A \otimes B)(\bx,\by,z)
\end{align*}
  The second case, where $\bx' \geq \bx$ and $z' > z$ is similar.
\end{proof}
\begin{lemma}
  $\upper(A \otimes B)$ is strictly convex.
\end{lemma}
\begin{proof}
   Again it is enough to show $h$ is strictly convex.
   Pick two different points $(\bx,\by)$ and $(\bx',\by')$.
  Recall that $f$ and $g$ are strictly convex by hypothesis.
  For $t \in (0,1)$,
  \begin{multline*}
      f((1-t)(\bx,\by) + t(\bx',\by'))\\
      < (1-t)f(\bx,\by) + t f(\bx',\by')
  \end{multline*}
  \begin{multline*}
      b+f(a)-f((1-t)(\bx,\by) + t (\bx',\by'))\\
      > (1-t)(b+f(a)-f(\bx,\by))\\ + t(b+f(a)-f(\bx',\by'))
  \end{multline*}
  \begin{multline*}
    g(b+f(a)-f((1-t)(\bx,\by) + t(\bx',\by'))) \\
    < (1-t)g((b+f(a)-f(\bx,\by))\\ + tg(b+f(a)-f(\bx',\by'))
  \end{multline*}
  \begin{multline*}
    h((1-t)(\bx,\by) + t(\bx',\by')) \\
    < (1-t)h(\bx,\by) + t h(\bx',\by')
  \end{multline*}
\end{proof}
Here is how stable points behave under sequential composition of multi-dimensional AMMs.

\begin{theorem}
\sloppy
  Let $(\bv,w,\bv')$ be a valuation, $\bv \in \Reals^n$, 
  $w \in \Reals$, and
  $\bv' \in \Reals^m$.
  If $(\ba^{*},f(\ba^{*}))$ is the stable point on AMM $(\bx,f(\bx))$
  for valuation $(\bv,w) / \|(\bv,w)\|_1$ and
  $(c,\bb^{*},g(c,\bb^{*}))$ is the stable point on AMM $(z,\by,g(z,\by))$
  for  valuation $(w,\bv') / \|(w,\bv')\|_1$,
  then $(\ba^{*},\bb^{*},h(\ba^{*},\bb^{*}))$ is the stable point for the valuation
  $(\bv,\bv') / \|(\bv,\bv')\|_1$.
  \end{theorem}

\begin{proof}
if $(\ba^{*},\bb^{*},h(\ba^{*},\bb^{*}))$ is not a stable point for $(\bv,\bv')$,
there is some $(\bx,\by) \neq (\ba^{*},\bb^{*})$ such that
$(\bx,\by,h(\bx,\by))$ is the stable point.
We write $\bv' = (\bv'_{m-1},v'_m)$ to separate the first $m-1$ components from the $m$-th component.
By assumptionm for $\bx \neq \ba^{*}$ and $(z,\by) \neq (c,\bb^{*})$,
\begin{align*}
    bv \cdot \ba^{*} + wf(\ba^{*}) 
    &< \bv \cdot \bx + wf(\bx) wc + \bv'_{m-1} \cdot \bb +  v'_m g(c,\bb) \\
    &< wz + \bv'_{m-1} \cdot \by +  v'_m g(z,\by)
\end{align*}
Also by assumption
\begin{align*}
  \bv \cdot \bx& + \bv'_{n-1} \cdot \by + v'_{m}h(\bx,\by)\\
  &< \bv \cdot \ba^{*} + \bv'_{m-1} \cdot \bb^{*} + v'_{m}h(\ba^{*},\bb^{*}) \\
  &= \bv \cdot \ba^{*} +wf(\ba^{*}) - wf(\ba^{*}) + wc - wc +  \bv'_{m-1} \cdot \bb^{*} \\
  &\quad \quad + v'_{m}g(c,\bb^{*}) \\
    &< \bv \cdot \bx + wf(\bx) - wf(\ba^{*}) + wz +\\
    &\quad \quad\bv'_{m-1} \cdot \by + v'_mg(z,\by) - wc \\
    &= \bv \cdot \bx + \bv'_{m-1} \cdot \by + v'_mg(z,\by)\\ 
    &= \bv \cdot \bx + \bv'_{m-1} \cdot \by + v'_mg(c + f(\ba^{*}) - f(\bx),\by) \\
    &= \bv \cdot \bx + \bv'_{m-1} \cdot \by + v'_mh(\bx,\by)
\end{align*}
a contradiction.

\end{proof}

\subsection{Many-to-Many Composition}
What does it mean to compose two AMMs that share multiple assets?
We will see that this definition requires some care.

Here is the most obvious definition.
Let $W,X,Y,Z$ be asset types.
Consider two 3-dimensional AMMs, $A$ defined by $w x y = 1$ and 
$B$ defined by $x y z = 8$.
where $A$ is in state $(1,1,1)$ and $B$ in state $(2,2,2)$.
We want to compose them into a 2-dimensional AMM $(A \oplus B)$ between $W$ and $Z$
treating $X$ and $Y$ as ``hidden'' intermediate assets.

A trader adds $dx > 0$ units of $X$ to $(A \oplus B)$:
\begin{enumerate}
\item pick any $dy,dz \leq 0$ satisfying $(1+dx)(1+dy)(1+dz) = 1$,
yielding a profit-loss vector on $Y,Z$ of $(dy,dz)$.
\item Subtract this profit-loss vector from the first two components of $B$,
then solve for $dw \leq 0$ so that $(2-dy) (2-dz) (2+dw) = 8$.
\end{enumerate}
More generally,
let $A$ be initialized in $(a,b,c)$, $B$ in $(b',c',d)$.
and $A \oplus B$ in $(a,d)$.
Then $(w,z)$ is in $A \oplus B$ if there exist $dx,dy \leq 0$
(the assets transferred)
such that $(w, b+dx, c+dy)$ is in $A$ and $(b'-dx, c'-dy, z)$ is in $B$.

Here is why this na\"ive definition of composition is flawed.
As before, $A(w,x,y) := wxy= 1$ starts in state $(1,1,1)$, 
$B(x,y,z) := xyz = 8$ starts in $(2,2,2)$.
Let $dw = 3$, $dx= -\frac{1}{2}$, $dy= - \frac{1}{2}$, so $dz$ satisfies:
\begin{equation*}
  (2-dx)(2-dy)(2+dz) = \frac{25}{4}(2+dz) = 2 
\end{equation*}
implying $dz = -\frac{42}{25}$.
If instead $dw = 3$, $dx= -\frac{1}{4}$, $dy= - \frac{2}{3}$, then $dz$ satisfies:
\begin{equation*}
  (2-dx)(2-dy)(2+dz) = 6(2+dz) = 2 
\end{equation*}
meaning $dz = -5/3$.
So both $(4,\frac{8}{25})$ and $(4,\frac{1}{3})$ are valid states in the na\"ive composition,
violating the requirement that each coordinate is a function of the others.

We have just seen that AMMs joined by more than one hidden asset type
provide too many degrees of freedom to allow composition in a simple, well-defined way.
In a practical sense, however,
if two AMMs agree on a valuation for the hidden assets,
then it makes sense to transfer them in proportion
to their agreed-upon relative values.
To define composition for AMMs with multiple hidden assets,
we create a single virtual asset from a convex combination of
the hidden assets (\secref{virtualization}).
Although any convex combination would produce a well-behaved composition,
it makes sense to use the current market valuation,
if one exists.
By reducing the number of hidden assets to one,
virtualization reduces many-to-many composition to many-to-one composition, as analyzed in \secref{many-to-one}.

Our next theorem concerns stable points under this most general type of sequential composition.
\begin{theorem}
\sloppy
Let $(\bv,\bw,\bv')$ be a valuation, $\bv \in \Reals^n,\bw \in \Reals^k,\bv' \in \Reals^m$.
Let $(\ba^{*},\bb^{*})$ be the stable state for
valuation $(\bv,\bw) / \|(\bv,\bw)\|_1$ on the AMM $A(\bx,\by)$
and $(\bc^{*},\bd^{*},e^{*})$ the stable state for
$(\bw,\bv') / \|(\bw,\bv')\|_1$ on the AMM $B(\bz,\br,q)$.
If $A|\bw:=(\bx,f(\bx))$ and $B| \bw:=(z,\br,g(z,\br))$,
then $(\ba^{*},\bd^{*},h(\ba^{*},\bd^{*}))$,
where $e^{*} = h(\ba^{*},\bd^{*})$ is the stable state on $(A | \bw) \oplus (B | \bw)$
for the valuation $(\bv,\bv')  / \|(\bv,\bv')\|_1$.
\end{theorem}

\begin{proof}
Applying \lemmaref{virtual-stable} twice,
$(\ba^{*},f(\ba^{*}))$ is the stable point for valuation
$(\bv,\|\bw\|_2^2) / \|(\bv,\|\bw\|_2^2)\|_1$
and $(t^{*},\bd^{*},g(t^{*},\bd^{*}))$ is the stable point for valuation
$(\|(\bw\|_2^2,\bv') / \|(\|(\bw\|_2^2,\bv')\|_1$,
where $e^{*} = g(t^{*},\bd^{*})$.
Applying \lemmaref{virtual-stable},
$(\ba^{*},\bd^{*},h(\ba^{*},\bd^{*}))$ is stable on $(A | \bw) \oplus (B | \bw)$ with respect to
valuation \\ $(\bv,\bv')/ \|(\bv,\bv')\|_1$.
   
\end{proof}
This result shows that the composition of AMMs at stable points remains stable.

\section{Parallel Composition}
\seclabel{parallel}

Parallel composition arises when a trader is faced with multiple AMMs,
but wants to treat them them as if they were a single AMM.
In sequential composition,
the composed AMMs exchange ``hidden'' assets .
In parallel composition,
the composed AMMs compete for overlapping assets.

Suppose Alice wants to trade asset $X$ for asset $Y$.
Bob and Carol both offer AMMs to convert from $X$ to $Y$.
Bob's AMM is $B(x,y):=x^2y=\frac{3}{4}$ in state $(1,\frac{3}{4})$,
while Carol's AMM is $C(x,y):=x y = 1$ in state $(1,1)$.
Alice would like to compose the two AMMs and treat them as one AMM.
Bob provides a better initial rate of exchange for small trades,
but Carol provides less slippage for large trades.
One can check that if Alice converts 1 unit of $X$,
she gets more $Y$ assets from Bob than from Carol,
while if she converts 3 units,
she gets more from Carol.

This process is not the same as order-book clearing,
because order-book offers are typically expressed in terms of a fixed amount and a fixed price,
while parallel AMM offers are expressed in terms of price curves.
This type of composition occurs, for example,
when a trader is faced with multiple pools as in Uniswap v3~\cite{uniswapv3}.

A rational Alice will split her assets between
Bob and Carol to maximize her return.
Suppose $A(x,y)$ is in state $(a,f(a))$
and $B(x,y)$ in $(b,g(b))$.
Alice splits her $d$ assets,
transferring $t d$ to Bob's AMM
and $(1-t)d$ to Carol's,
returning
\begin{equation*}
  f(a) - f(a+tx) + g(b) - g(b+(1-t)x)
\end{equation*}
units of $Y$.
Let $h(x) = f(a+tx) + g(b+(1-t)x)$.
Define the \emph{parallel composition} of $B$ and $C$
with respect to $v=(t,1-t)$ to be.
\begin{equation*}
    (B\|C)(x,y) := y - h(x) = 0.
\end{equation*}
\begin{lemma}
  $(B\|C)(x,y)$ is a 2-dimensional AMM.
\end{lemma}
\begin{proof}
  $(B\|C)(x,y)$ is twice-differentiable because $f$ and $g$ are twice-differentiable.
  To check that $(B\|C)(x,y)$ is strictly increasing,
  let $x' \geq x$ and $y' \geq y$ where at least one inequality is strict.
  For the first case, suppose $x' > x$ and $y' \geq y$.
  \begin{align*}
    f(a+tx') &< f(a+tx) \\
    g(b+(1-t)x') &< g(b+(1-t)x) \\
    f(a+tx') + g(b+(1-t)x') &< f(a+tx) + g(b+(1-t)x) \\
    h(x') &< h(x) \\
    y - h(x') &> y-h(x)\\
    y' - h(x') &> y-h(x)
  \end{align*}
  The case where $x' \geq x$ and $y'> y$ is similar.

  To check that $\upper((B\|C)(x,y))$ is strictly convex,
  we can verify $h$ is strictly convex.
  Pick distinct $x,x'$.
  For $s \in (0,1)$,
  \begin{multline*}
    s f(a+t x) + (1-s) f(a+t x')\\
    > f(s (a+t x) + (1-s) (a+t x')) 
  \end{multline*}
  \begin{multline*}
    s g(b+(1-t)x) + (1-s) g(b+(1-t)x') \\
    >g(s (b+(1-t)x) + (1-s) (b+(1-t)x'))
  \end{multline*}
  \begin{multline*}
    s (f(a+t x) + g(b+(1-t)x))\\ + (1-s)(f(a + t x')+g(b+(1-t)x'))\\
    > f(s (a + t x) + (1-s) (a + t x'))\\ + 
    g(s (b+(1-t)x) + (1-s) (b+(1-t)x'))
  \end{multline*}
  \begin{equation*}
    s h(x) + (1-s)h(x') > h(s x + (1-s)x')  \\
  \end{equation*}
  which establishes the claim.
\end{proof}

\begin{lemma}
  Let $A := (x,f(x))$, $B := (y,g(y))$ be two AMMs trading assets $X$ and $Y$,
  such that $(a,f(a))$ and $(b,g(b))$ are their respective stable points
  for the valuation $(v,1-v)$.
  If $h_t(x) = f(a + tx) + g(b + (1-t)x)$,
  then $(0,h_t(0))$ is stable point on $A \| B$ with respect to $(v,1-v)$ for all $t \in \Reals$.
\end{lemma}

\begin{proof}
By assumption we have
\begin{align*}
    v a + (1-v)f(a) &< v(a + tx) + (1-v)f(a + tx)\\
    (1-v)f(a) &< tvx + (1-v)f(a + tx)
\end{align*}
and
\begin{align*}
  vb + (1-v)g(b) &< v(b + (1-t)x) + (1-v)g(b + (1-t)x)\\
  (1-v)g(b) &< (1-t)vx + (1-v)g(b + (1-t)x),
\end{align*}
yielding
\begin{align*}
  v 0 + (1-v)h_t(0)
  &= (1-v)h_t(0) \\
  &= (1-v)f(a) + (1-v)g(b) \\
  &< tvx + (1-v)f(a + tx) + (1-t)vx \\
  &\quad \quad + (1-v)g(b + (1-t)x) \\
  &= vx + (1-v)(f(a + tx) + g(b + (1-t)x)) \\
  &= vx + (1-v)h(x)
\end{align*}
so $(0,h_t(0))$ is a stable point for $(v,(1-v))$.
\end{proof}

Parallel composition is well-defined for any valuation,
but what valuation should a rational Alice pick?
Differentiating with respect to $t$ yields
\begin{align}
  0 &= -x f'(a+t x) + x g'(b+(1-t)x)\nonumber\\
  x f'(a+t x) &= x g'(b+(1-t)x) \nonumber\\
  f'(a+t x) &= g'(b+(1-t)x).\eqnlabel{best-split}
\end{align}
Alice maximizes her return when she splits her
assets so that Bob and Carol end up offering the same rate.
Informally, if Bob had ended up providing a better rate,
then Alice should have given him a larger share.
If there is no $t \in (0,1)$ that satisfies \eqnref{best-split},
Alice should give all her assets to the AMM with the better rate.

How should parallel composition be defined for AMMs with multiple asset types?
Suppose Alice has some combination $\bd$ of assets in $X_1,\ldots,X_p$
that she wants to convert into some combination of assets in $Y_1,\ldots,Y_q$.
Alice has a choice of two alternative AMMs:
$B(x_1,\ldots,x_p,y_1,\ldots,y_q)$ and $C(x_1,\ldots,x_p,y_1,\ldots,y_q)$.
Perhaps the most sensible way to define parallel composition is through
asset virtualization.
Alice's input asset vector $\bd$ induces a valuation $\bd/ \|\bd\|_1$
which can be used to define a virtual asset $X$ from $X_1,\ldots,X_p$.
Alice chooses a valuation $\bv$ for her $Y_1,\ldots,Y_q$ outputs
(perhaps the market valuation)
which can be used to define a virtual asset $Y$ from $Y_1,\ldots,Y_q$.
After asset virtualization, the alternative AMMs have the form:
$\tilde{B}(\bx,t)$ and $\tilde{C}(\bx,t)$,
and the definition of parallel composition proceeds as before.

The next lemma describes properties of stable points for two 2-dimensional AMMs composed in parallel.

If $(\bx,f(\bx))$ and $(\by,g(\by))$ are two $n$-dimensional AMMs in states $(\ba,f(\ba))$ and $(\bb,g(\bb))$,
we can define parallel composition as follows.
For $\bt \in [0,1]^{n-1}$,
let $h_{\bt}(\bx) = f(\ba + \bt * \bx) + g(\bb + (\bone - \bt )* \bx)$,
where $*$ is component-wise multiplication.

\begin{lemma}
  \lemmalabel{parallel-stable}
  Let $(\bv,1-\|\bv\|_1))$ be a valuation,
  and let $A := (\bx,f(\bx))$ and $B := (\by,g(\by))$.
  If $(\ba,f(\ba))$ and $(\bb,g(\bb))$ are both stable points with respect to $(\bv,1-\|\bv\|_1))$.
  then $(\bzero,h_{\bt}(\bzero))$ is the stable point on $A \| B$  for $(\bv,1-\|\bv\|_1)$
  for all $\bt \in \Reals^{n-1}$.
  \end{lemma}

  \begin{proof}
    For $\bt \in \Reals^{n-1}$,
    \begin{multline*}
      \bv \cdot \ba + (1-\|\bv\|_1)f(\ba)\\
      < \bv \cdot (\ba + \bt * \bx) + (1-\|\bv\|_1)f(\ba + \bt * \bx)
    \end{multline*}
    \begin{equation*}
      (1-\|\bv\|_1)f(\ba)
      < (\bt * \bx) \cdot \bv + (1-\|\bv\|_1)f(\ba + \bt * \bx)
    \end{equation*}
    and
    \begin{multline*}
      \bv \cdot \bb + (1-\|\bv\|_1)g(\bb)\\
      < \bv \cdot (\bb + (\bone - \bt) * \by) + (1-\|\bv\|_1)g(\bb + (\bone - \bt) * \by)
    \end{multline*}
    \begin{multline*}
      (1-\|\bv\|_1)g(\bb) \\
      < ((\bone - \bt) * \by) \cdot \bv + (1-\|\bv\|_1)g(\bb + (\bone - \bt) * \by),
     \end{multline*}
implying
\begin{align*}
  \bv \cdot \bzero + &(1-\|\bv\|_1)h_t(\bzero) \\
  &= (1-\|\bv\|_1)h_t(\bzero) \\
  &= (1-\|\bv\|_1)f(\ba) +(1-\|\bv\|_1)g(\bb) \\
  &< (\bt * \bx) \cdot \bv + (1-\|\bv\|_1)f(\ba + \bt * \bx) \\
  &\quad \quad+ (\bone-\bt) * \bx \cdot \bv + \\
  &\quad \quad (1-\|\bv\|_1)g(\bb + (\bone-\bt) * \bx)\\
    &= \bv \cdot \bx + (1-\|\bv\|_1)(f(\ba + \bt \cdot \bx) \\
    &\quad \quad + g(\bb + (\bone-\bt) \cdot \bx))  \\
    &= \bv \cdot \bx + (1-\|\bv\|_1)h(\bx)
\end{align*}
so $(0,h_t(0))$ is a stable point for $(\bv,1-\|\bv\|_1)$.
\end{proof}
Most generally, if we have two AMMs $A(\bx,\bz)$ and $B(\by,\bz')$, and valuation $\bw$,
we can write $A| \bw$ as $(\bx,f(\bx))$ and $B| \bw$ as $(\by,g(\by))$.
We then can define parallel composition as before.

\begin{theorem}
    Let $(\bv,\bv')$ be a valuation, and let $A(\bx,\bz)$ and $B(\by,\bz')$ two AMMs.
    If $(\ba^{*},\bb^{*})$ is the stable point for $A$,
    and $(\bc^{*},\bd^{*})$ the stable point for $B$,
    both with respect to $(\bv,\bv')$,
    then $(\bv,\|\bv'\|_2^2)/\|(\bv,\|\bv'\|_2^2)\|_1$ is the stable point for $(A | \bv' ) \| (B | \bv')$.
\end{theorem}

\begin{proof}\sloppy
  \lemmaref{virtual-stable} implies that $(\ba^{*},f(\ba^{*}))$ and $(\bc^{*},g(\bc^{*}))$
  are stable on $A | \bv'$ and $B | \bv'$,
  both with respect to $(\bv,\|\bv'\|_2^2)$. 
\lemmaref{parallel-stable} implies that $(\bzero,h_{\bt}(\bzero))$
is the stable point for
$(A | \bv' ) \| (B | \bv')$ with respect to valuation $(\bv,\|\bv'\|_2^2)$.
\end{proof}

\section{Fees}
\seclabel{fees}
In practice, each AMM trade incurs a \emph{fee}:
each trader's deposit includes a small additional fee
added directly to the AMM's capitalization to benefit the liquidity providers.
So far, we have neglected fees because the amounts involved
are expected to be small in relation to AMM capitalization.
Nevertheless, in this section,
we show that AMM fees can be modeled as the sequential composition
of a no-fee AMM with a simple \emph{linear AMM}.
For brevity, we restrict our attention to AMMs that manage two assets.

Here is how fees work for an AMM such as Uniswap v1.
Let $A := (x,g(x))$, currently at state $(a,g(a))$,
and let $\gamma, 0 < \gamma < 1$, be the fee parameter.
\begin{itemize}
\item
  The trader sends $\delta$ units of $X$ to $A$.
\item
  $(1-\gamma)\delta$ units of $X$ are traded for $g(a) - g(a +(1-\gamma)\delta)$ of $Y$.
\item
  A fee of $\gamma \delta$ units of $X$ is deposited directly in $A$'s pool.
\end{itemize}
$A$'s final state is $(a + \delta, g(a + (1-\gamma)\delta))$.

A \emph{linear AMM} exchanges assets at a constant rate,
governed by a constraint function:
\begin{equation*}
  \lambda \cdot \bx = c,
\end{equation*}
where $\lambda$ is a constant vector, and $c > 0$ a constant.
A linear AMM \emph{does not} satisfy all our common-sense axioms:
because its exchange rate is fixed, it is not expressive,
and although its curve is convex, it is not strictly convex,
so it does not have unique stable points.
A stand-alone linear AMM $L$ would be a poor investment for providers,
because if the market rate diverges from the AMM's fixed rate,
perhaps so that $L$ overprices $X$ and underprices $Y$,
then arbitrage traders will exchange $X$ for $Y$ until $L$'s $Y$ reserve is depleted.

Even if stand-alone linear AMMs are not useful in practice,
they provide a convenient formal device for modeling AMMs with fees.
For example,
consider the 2-dimensional AMM $L := (x,f(x))$ that trades between $X$ and $X$,
where
\begin{equation*}
f(x) = (1-\gamma)(a + \delta - x),  
\end{equation*}
for $\gamma, 0 < \gamma < 1$.
Note that
\begin{equation}
\eqnlabel{fee-f}
f(a) = (1-\gamma) \delta \qquad \text{and} \qquad f(a+\delta) = 0. 
\end{equation}
Suppose $A:=(x,g(x))$ starts in state $(a,g(a))$,
and linear $L:=(x,f(x))$ starts in state $(a,f(a))$.
As in \secref{sequential},
the sequential composition $L \otimes A$ is given by
\begin{equation}
\eqnlabel{fee-h}
h(x) = g(a + f(a) - f(a+x)).
\end{equation}
Sending $\delta$ units of $X$ to $L \otimes A$ returns
\begin{align*}
\eqnlabel{fee-delta}
  h(\delta)
  &= g(a + f(a) - f(a + \delta))\\
  &= g(a + (1-\gamma) \delta)
\end{align*}
The trade leaves $L \otimes A$ in state $(a + \delta, g(a + (1-\gamma)\delta))$,
precisely the behavior of $A$ augmented by a fee $\gamma$ levied on incoming assets.
An alternative structure where a fee is levied on outgoing assets
can be modeled by composing the AMMs in the reverse order,
with a linear AMM deducting $\delta \gamma$ units of the trade's output to its pool.

\section{Related Work}
\seclabel{related}
\balance
Angeris and Chitra~\cite{AngerisC2020}
introduce a \emph{constant function market maker} model
and consider conditions that ensure that agents who
interact with AMMs correctly report asset prices.
Our work, based on a similar but not identical AMM model,
focuses on properties such as defining AMM composition,
AMM topology, and the role of stable points.

\emph{Uniswap}~\cite{uniswap,zhang2018}
is a family of constant-product AMMs that originally traded between
ERC-20 tokens~\cite{erc20} and ether cryptocurrency.
Trading between ERC-20 assets requires sequential composition
of the kind analyzed in \secref{sequential}.
Uniswap v2~\cite{uniswapv2} added direct trading between
selected pairs of ERC-20 tokens,
and Uniswap v3 ~\cite{uniswapv2} allows liquidity providers
to restrict the range of prices in which their asset participate,
giving rise to a form of parallel composition of the kind analyzed in \secref{parallel}.

\emph{Bancor}~\cite{bancor} AMMs permit more flexible pricing schemes.
The state space manifold is parameterized by a \emph{weight},
where different weights yield different curves.
Bancor AMMs trade between ERC-20 assets and Bancor-issued BNT tokens.
Prices are a function of assets held and BNT tokens in circulation. 
Later versions~\cite{bancorv2} include integration with external
``price oracles'' to keep prices in line with market conditions.

\emph{Balancer}~\cite{balancer} AMMs trade across more than two assets.
Instead of constant product,
their state space is given by \emph{constant mean} formula
$c = \Pi_i^n x_i^{w_i}$,
where $c$ is constant, $x_i$ is amount of asset $X_i$ held by the contract,
and the $w_i$ are adjustable weights that form a valuation.

\emph{Curve}~\cite{curve} uses a custom curve specialized for trading across
multiple \emph{stablecoins},
digital assets whose values are tied to fiat currencies such as the
dollar, and likely to trade at near-parity.

There are many are more examples of AMMs:
see Pourpouneh \emph{et al.}~\cite{pourpouneh} for a survey.

Before there were AMMs for decentralized finance,
there were AMMs for \emph{event prediction markets},
where parties trade securities that pay a premium if
and only if some event occurs within a specified time.
A community of researchers has focused on prediction-market AMMs
\cite{AbernethyYV2011,ChenW2010,ChenP2007,Hanson2003,Hanson2007}.
Despite superficial similarities,
event prediction AMMs and security AMMs differ from DEFI AMMs
is important ways:
pricing models are different because
prediction outcome spaces are discrete rather than continuous,
prediction securities have finite lifetimes,
and composition of AMMs is not a concern.

We also note that the mathematical structure of AMMs resembles that of a consumer utility curve from classical economics ~\cite{micro}, where assets are replaced with goods.
The optimal arbitrage problem is not new.
A consumer choosing an optimal bundle of goods for a fixed set of prices is the same as an arbitrageur choosing an optimal point on an AMM with respect to a market valuation.
This is known as the \emph{expenditure minimization problem}~\cite{micro}.
While AMMs and consumer indifference surfaces are mathematically similar, they are different in application.
In particular,
traders interact with AMMs via composition,
an issue that does not arise in the consumer model.

\section{Conclusions}
\seclabel{conclusions}
Modern AMMs provide increasingly complex rules for trades.
For example, Uniswap v3~\cite{uniswapv3} allows liquidity providers to
choose to take trades only over finite ranges.
In the future,
we would like to understand how to define composition operators
for varieties of range-restricted AMMs,
but we hope the reader is convinced that understanding composition
of simple AMMs is already challenging,
and an important step to understand more general cases.

AMMs are increasingly being integrated with external price
information oracles.
For example,
Krishnamachari \emph{et al.}~\cite{krishnamachari2021dynamic}
describe a family of AMMs capable of adjusting their
curves in response to reported price changes.
In future work,
we plan to investigate composition for AMMs that make use of oracle services..

The work presented here is a first step toward analyzing
decentralized finance from a distributed computing perspective.
In future work,
we hope to consider more complex networks of AMMs,
the challenges of cross-chain AMMs,
synchronization problems such as front-running,
as well as more adaptive AMM structures.

\begin{acks}
This research was supported by NSF grant 1917990.
The authors are grateful to the conference referees for helpful remarks.
\end{acks}
\newpage
\bibliographystyle{ACM-Reference-Format}
\bibliography{references}

\end{document}